\documentclass[12pt, letterpaper, twoside]{article}

\usepackage[legalpaper, letterpaper]{geometry}
\usepackage{hyperref}
\usepackage{comment}
\usepackage{amsmath,amsfonts,theorem,epsfig,amssymb, mathtools}
\usepackage{float}
\newtheorem{prop}{Proposition}
\newtheorem{theorem}{Theorem}
\newtheorem{corollary}{Corollary}
\newtheorem{lemma}{Lemma}
\newtheorem{remark}{Remark}
\newcommand\myeq{\stackrel{\mathclap{\normalfont\mbox{def}}}{=}}
\newenvironment{proof}{\paragraph{Proof:}}{\hfill$\square$}

\usepackage{graphicx}

\usepackage[]{graphicx}
\usepackage{caption}
\usepackage{subcaption}
\usepackage{authblk}

\usepackage[utf8]{inputenc}
\usepackage{biblatex}
\providecommand{\keywords}[1]
{
  \small	
  \textbf{\textit{Keywords---}} #1
}

\title{Theoretical Foundation of Colored Petri Net through an Analysis of their Markings as Multi-classification}

\makeatletter
\newcommand\email[2][]%
   {\newaffiltrue\let\AB@blk@and\AB@pand
      \if\relax#1\relax\def\AB@note{\AB@thenote}\else\def\AB@note{\relax}%
        \setcounter{Maxaffil}{0}\fi
      \begingroup
        \let\protect\@unexpandable@protect
        \def\thanks{\protect\thanks}\def\footnote{\protect\footnote}%
        \@temptokena=\expandafter{\AB@authors}%
        {\def\\{\protect\\\protect\Affilfont}\xdef\AB@temp{#2}}%
         \xdef\AB@authors{\the\@temptokena\AB@las\AB@au@str
         \protect\\[\affilsep]\protect\Affilfont\AB@temp}%
         \gdef\AB@las{}\gdef\AB@au@str{}%
        {\def\\{, \ignorespaces}\xdef\AB@temp{#2}}%
        \@temptokena=\expandafter{\AB@affillist}%
        \xdef\AB@affillist{\the\@temptokena \AB@affilsep
          \AB@affilnote{}\protect\Affilfont\AB@temp}%
      \endgroup
       \let\AB@affilsep\AB@affilsepx
}
\makeatother

\author[]{Jules Chenou }
\author[]{George Hsieh}
\author[]{Aurelia Williams}
\affil[]{Cybersecurity Complex, Center of Excellence of Cybersecurity, Norfolk State University}
\email{\url{{jchenou,ghsieh, atwilliams}@nsu.edu}}
\date{}

\begin{document}

\maketitle
\begin{abstract}
Barwise and Seligman stated the first principle of information flow: "Information flow results from regularities in the distributed system."  They represent a distributed system in terms of a \textit{classification} consisting of a set of objects or tokens to be classified, a set of types used to classify tokens, and a binary relation between tokens and types that tells one which tokens are classified as being of which types. We aim to further this investigation and proceed with a dynamic or evolving system instead of a static system.\\
We claim that a classification is a snapshot of a distributed system at a given moment or context. We then aim to answer the question posed by an evolving context. As the context or configuration changes, how do regularities evolve.\\
This paper is a continuation of an investigation we started in \cite{esterlin},  where we initiated how to capture a dynamism of information flow with a Kripke structure.   Here we develop the same procedure with colored  Petri net(CPN). We first extend the classification concept to multiclassification by replacing its binary relation between tokens and types with a multi relation: a function from \(tok\mathcal(A)\) \(\times\) \(typ\mathcal(A)\) to \(\mathbb{N}\), the set of natural numbers. The multiclassification will unfold into binary classification in order to compute its theory. It turns out that markings of a CPN are multiclassification;  Amalgamating the theories of those classifications obtained as markings of CPN results in a CPN's knowledge base. 
\end{abstract}
\section{Introduction}
Barwise and Seligman, in \cite{barwise}, apply category-theoretic notions to the problem of information flow.  This work made use of notions of situation that Barwise developed in earlier work (addressing the situation in logic and situation semantics).\par
Information flow provides
\begin{itemize}
\item
a framework of how to represent information (classifications),
\item
reasoning about information (theories), and
\item
gluing information into a consistent and coherent model (local logic).
\end{itemize}
Barwise’s approach is missing something inherent to physical systems: the dynamism, in other words, how do classifications representing an agent‘s view of things evolve as the situation changes. The local logic will not remain the same, but how to capture the evolution of theories from situations. The same question arises if one looks for the variation on the classification side.\par

There are two alternative solutions to these questions: The external way consists of providing a classification with some dynamism; by doing so, we obtain a structure similar to a Kripke structure.  Thus, a Kripke structure is a Barwise classification augmented with an accessibility relation on the set of tokens. [temporal-spatial relation] \par
The other way will consist of imagining a mechanism for the evolution of type as a set of tokens. It turns out that one such mechanism is the enabling condition of colored Petri net, with classifications, as markings. \par
This paper is organized as follows: In section 2, we recall the basics of classification and infomorphism. We give a characterization of infomorphism as a commutative diagram and expand the notion of classification to multi-classification, where type is now viewed as a multiset of tokens instead of a set of tokens. In section 3, we will describe how to obtain the state graph of a colored Petri net as evolving multi-classifications. Combining all the theories obtained from multi-classification gives us the knowledge base of a colored Petri net. Section 4 is the conclusion and direction for future research.
\section{Basic of Classification and Channel Theory}
Barwise and Seligman \cite{barwise} presented a framework for the “flow of information” in (generally implicitly) category-theoretic terms. They address the question, “How it is that information about some component of a system carries information about other components of the system?” \par
They define a classification \( \mathcal{A}\) to be a structure with non-empty sets typ(\( \mathcal{A}\)) of types and tok(\( \mathcal{A}\)) of tokens as well as a binary relation \(\vDash_\mathcal{A}\) between tok(\( \mathcal{A}\)) and typ(\( \mathcal{A}\)) such that, for \(a \in tok(A) \)  and \(\alpha \in typ(A) \), \( a \vDash_\mathcal{A} \alpha \)    indicates that a is of type \( \alpha \). The theory does not limit what a or \( \alpha \)  might be (as long as it makes sense for a to be of type \( \alpha \).  It could be that a is an object and \(\alpha \) a property (monadic first-order relation), or a might be a situation and \( \alpha \) a type of situation; often, different tokens of a classification amount to the same physical system across different time points and types are instantaneous partial state descriptions of the system. \par
Giving two classifications \( \mathcal{A} \) and \( \mathcal{C} \), an infomorphism f from \( \mathcal{A} \) to \( \mathcal{C} \) is a pair of functions
\begin{equation} \label{eq1}
\begin{split}
(f^\wedge,f^\vee), f^\wedge: typ( \mathcal{A} )  \longrightarrow  typ( \mathcal{C} ) \; and \\
f^\vee:tok( \mathcal{C} )  \longrightarrow  tok( \mathcal{A} ) 
\end{split}
\end{equation}
satisfying, for all tokens \( c \in tok(\mathcal{C}) \) and all types \( \alpha \in typ(\mathcal{A}) \)
\begin{equation} \label{eq2}
f^\vee(c) \vDash_\mathcal{A} \alpha \;\; iff \;\; c  \vDash_\mathcal{C} f^\wedge(\alpha)
\end{equation}
Composition of binary relation: Let  \(R \subseteq A \times B\) and \(S \subseteq B \times C\)  such that \(codomain(R) = domain(S)\) then \(S \circ R = \{(a,c) \in A \times C  \vert \exists b \in B (a,b) \in R \wedge (b,c) \in S \}\).  \par
With the above definition of composition of binary relations, If \( (f^\wedge,f^\vee):\) (typ(\(\mathcal{A}\)), tok(\(\mathcal{A}\)), \(\vDash_\mathcal{A}\)) \(\rightarrow\) (typ(\(\mathcal{C}\)), tok(\(\mathcal{C}\)), \(\vDash_\mathcal{C}\)) is an informorphism, the composition \(f^\vee \circ \vDash_\mathcal{A} \circ f^\wedge \) is well defined. Indeed, \(f^\vee = \{(c,f^\vee(c) \forall c \in tok(\mathcal{A}) \subseteq tok(\mathcal{A}) \times tok(\mathcal{A})\). Furthermore, we have, \par

\begin{prop}
 If \( (f^\wedge,f^\vee): (typ(\mathcal{A}), tok(\mathcal{A}), \vDash_\mathcal{A}) \rightarrow (typ(\mathcal{C}), tok(\mathcal{C}), \vDash_\mathcal{C})\) is an informorphism  then \( f^\vee \circ \vDash_\mathcal{A} \circ f^\wedge \subseteq \hspace{0.2cm} \vDash_\mathcal{C}\)
 \end{prop}
\begin{proof}
 Let \((b,\beta) \in f^\vee \circ \vDash_\mathcal{C} \circ f^\wedge = f^\vee \circ (\vDash_\mathcal{C} \circ f^\wedge)\) \\
 \((b,\beta) \in f^\vee \circ (\vDash_\mathcal{C} \circ f^\wedge) \Leftrightarrow \exists a \in A \vert (b,a) \in f^\vee \wedge (a, \beta) \in (\vDash_\mathcal{C} \circ f^\wedge)\)\\
 \(\Leftrightarrow \exists a \in A \vert (b,a) \in f^\vee \wedge \exists \alpha \in typ(\mathcal{A}) \vert (a, \alpha) \in \vDash_\mathcal{C} \wedge (\alpha, \beta) \in f^\wedge \)\\
\(\Leftrightarrow \exists a \in A \vert a = f^\vee(b)  \wedge \exists \alpha \in typ(\mathcal{A}) \vert (a, \alpha) \in \vDash_\mathcal{C} \wedge \beta = f^\wedge(\alpha)\)\\
\(\Leftrightarrow \exists a \in A \vert a = f^\vee(b) \wedge \exists \alpha \in typ(\mathcal{A}) f^\vee(b) \vDash_\mathcal{C} \alpha \wedge \beta = f^\wedge(\alpha)\)\\
\(\Leftrightarrow \exists a \in A \vert a = f^\vee(b) \wedge \exists \alpha \in typ(\mathcal{A}) b \vDash_\mathcal{A} f^\wedge(\alpha) \wedge \beta = f^\wedge(\alpha)\)\\
\(\Leftrightarrow \exists a \in A \vert a = f^\vee(b) \wedge \exists \alpha \in typ(\mathcal{A}) b \vDash_\mathcal{A} \beta \wedge \beta = f^\wedge(\alpha)\)\\
\(\Rightarrow (b , \beta) \in \vDash_\mathcal{A}\).
\end{proof}

 \par
 In graphic terms it means the following diagram in Figure 1"commute" This was first drawed in \cite{gebreyohannes}.
\begin{figure}[ht!]
  \includegraphics{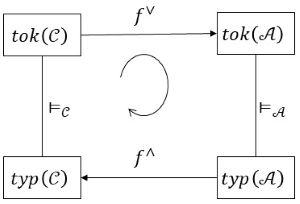}
  \caption{A commutative diagram of an infomorphism }
\end{figure}

Intuitively, an infomorphism is a part-to-whole, A-to-C, informational relationship. Even though the information is not defined, it is assumed to “flow” among the components of a system. Components may, but need not, be distant from one another in time and space, and they may be very different things.  The system is “distributed” in this sense (not necessarily in the sense in which that term is used in computer science).  For example, the students, classrooms, scheduling system, and attendance records together form a distributed system. \par
Turning to regularities in a classification’s types, let \(\mathcal{A}\) be a classification and \(\Gamma\) and \(\Delta\) be sets of types in A.  A token a of A satisfies the “sequent” \(\langle\Gamma,\Delta\rangle\), provided that, if a is of every type in \(\Gamma\), then it is of some type in \(\Delta\); which we denote \textit{a} \(\Vdash\) \(\langle\Gamma,\Delta\rangle\).\\
We can observe that 
\( \textit{a} \Vdash \langle \Gamma,\Delta \rangle \myeq  \Gamma \subseteq \hat{a} \Rightarrow \Delta \cap  \hat{a} \neq \emptyset\).
where \(\hat{a}\) = \(\{\alpha\) \(\in\) typ(\(\mathcal{A}\)) \(\vert\) \textit{a}\(\vDash_\mathcal{A}\) \(\alpha\}\)\\
If every token of \(\mathcal{A}\) satisfies \(\langle\Gamma,\Delta\rangle\), then \(\Gamma\) is said to entail \(\Delta\) and \(\langle\Gamma,\Delta\rangle\) is called a constraint supported by \(\mathcal{A}\).  The set of all constraints supported by \(\mathcal{A}\) is called the complete theory of \(\mathcal{A}\), denoted by \(Th(\mathcal{A})\). \par
These constraints are system regularities, and it is by virtue of regularities among connections that information about some components of a distributed system carries information about other components.  These regularities are relative to the analysis of the distributed system in terms of information channels. Barwise and Seligman’s summary statement of their analysis of information flow, restricted to the simple case of a system with two components, a and b, is as follows.

\begin{remark}
Given two sequents \(\langle \Gamma_1, \Delta_1 \rangle \) and \(\langle \Gamma_2, \Delta_2 \rangle \), the relation defined by \(\langle \Gamma_1, \Delta_1 \rangle \) \(\sqsubseteq\) \(\langle \Gamma_2, \Delta_2 \rangle \) if and only if \(\Gamma_1 \subseteq \Gamma_2 \) and \(\Delta_1 \subseteq \Delta_2\) is an order (its is reflexive, antisymetric and transitive)
\end{remark}

\begin{prop}
Let \(\textit{a} \in \textit{tok}(\mathcal{A}) a token of a classification \mathcal{A},  \langle \Gamma_1, \Delta_1 \rangle  and \langle \Gamma_2, \Delta_2 \rangle  two sequents of \mathcal{A}, if \langle \Gamma_1, \Delta_1 \rangle  \sqsubseteq \langle \Gamma_2, \Delta_2 \rangle  and \textit{a} \Vdash \langle\Gamma_1,\Delta_1\rangle then \textit{a} \Vdash \langle\Gamma_2,\Delta_2\rangle\)
\end{prop}
\begin{proof}
\textit{a} \(\Vdash \langle\Gamma_1,\Delta_1\rangle \Leftrightarrow
\Gamma_1 \subseteq \hat{a} \Rightarrow \hat{a} \cap \Delta_1 \neq \emptyset\) \\
If \(\Gamma_2 \subseteq \hat{a}\) it's straightforward that \(\hat{a} \cap \Delta_1 \neq \emptyset\): the same elements that validated \(\hat{a} \cap \Delta_1 \neq \emptyset\) will fufill for \( \hat{a} \cap \Delta_2 \neq \emptyset\) because \(\Delta_1 \subseteq \Delta_2 \)
\end{proof}
\begin{corollary}
If \(\langle \Gamma_1, \Delta_1 \rangle  \sqsubseteq \langle \Gamma_2, \Delta_2 \rangle\)  and \( (\langle \Gamma_1, \Delta_1 \rangle \)  is a constraint of \(\mathcal{A}\) then  \(\langle \Gamma_2, \Delta_2 \rangle\)  is a constraint of \(\mathcal{A}\) as well.
\end{corollary}
\subsection{Example of Theory of a classification}
Consider the classification  \(\mathcal{A}\) in Table 1 with tok(\(\mathcal{A}\))=\{a,b,c\}, typ(\(\mathcal{A}\))=\{\(\alpha\),\(\beta\),\(\delta\)\} and 
\(\vDash_\mathcal{A}\) = \{((a,\(\alpha\)),(a,\(\delta\)),(b,\(\delta\)),(b,\(\beta\)),(c,\(\beta\))\}. \par
\begin{table}[ht!]
    \centering
    \begin{tabular}{|c|c|c|c|c|c}
    \hline
      \(\mathcal{A}\)   & \(\alpha\) & \(\beta\) &  \(\delta\) \\
      \hline
     a    &  1 & 0 & 1 \\
     \hline
     b & 0 & 1 & 1 \\
     \hline
     c & 0 & 1 & 0 \\
     \hline
     
    \end{tabular}
    \caption{Example of classification}
    \label{tab:my_label}
\end{table}
The theory of \(\mathcal{A}\) (\cite{barwise}, page 124) is \(Th(\mathcal{A})=\{ \langle \alpha,\delta \rangle,\langle \emptyset,\{ \alpha,  \beta \} \rangle,\langle \{ \alpha,  \beta\},\emptyset\rangle\}\). \par
In view of Table 1, the only token of type \(\alpha\) is a, we can observe that a is also of type \(\delta\); here we remove the curly bracket around the singleton to make it readable. This explanation makes  \(\langle \alpha,\delta \rangle\) a sequent constraint supported by \(\mathcal{A}\). 

Given a classification \(\mathcal{A}\) and a set \(\Gamma\)  of types of \(\mathcal{A}\); let \(\bigcap \Gamma=\{x \in tok(\mathcal{A})| \forall \alpha \in \Gamma,x\vDash_\mathcal{A} \alpha\}\) is a subset of \(tok(\mathcal{A})\); with \(\bigcap \emptyset = tok(\mathcal{A})\) and alternatively, \(\bigcup \Gamma =\{x \in tok(\mathcal{A})|\exists \alpha \in \Gamma,x\vDash_\mathcal{A} \alpha\}\). The following proposition gives an algebraic perspective on the validity of a constraint in a classification. 

\begin{theorem} Given a classification \(\mathcal{A}\), and a sequent \(\langle\Gamma,\Delta\rangle\) in \(typ(\mathcal{A})\); \(\langle\Gamma,\Delta\rangle\) is a constraint of \(\mathcal{A}\) if and only if \(\bigcap \Gamma \subseteq \bigcup\Delta \).
\end{theorem}
\begin{lemma}
Giving two set of types \(\Gamma_1  and \Gamma_2 \)\\
\(\Gamma_1  \subseteq  \Gamma_2  \Rightarrow\)
\[ \begin{cases} 
       \bigcup \Gamma_1 \subseteq \bigcup \Gamma_2  \\
       \bigcap \Gamma_2 \subseteq \bigcap \Gamma_1 
   \end{cases}
\]
\end{lemma}
The proof of this lemma is straightforward.
\begin{proof} (Of theorem 1)
\(\Rightarrow is true because b \in \bigcap \Gamma \Rightarrow \Gamma \subseteq \hat{b}. If \langle\Gamma, \Delta\rangle is a constraint then by definition, \forall b \in tok(\mathcal{A}), \Gamma \subseteq \hat{b} \Rightarrow \Delta \cap \hat{b} \neq \emptyset. If  x \in \bigcap \Gamma  i.e. \Gamma \subseteq \hat{x} then \Delta \cap \hat{x} \neq \emptyset\).\\
\(\Delta \cap \hat{x} \neq \emptyset \Rightarrow \exists \alpha \in \Delta \vert x\vDash_\mathcal{A} \alpha\)\\
\(\Rightarrow x \in \bigcup \Delta\)\\
\(\Leftarrow\)) Inversely, suppose \(\bigcap \Gamma \subseteq \bigcup \Delta \).
Let \(a \in tok(\mathcal{A})\) such that \(\Gamma \subseteq \hat{a}\)\\
\(\Gamma \subseteq \hat{a} \Rightarrow\) \(\hat{a} \subseteq \cap \Gamma \subseteq \cup \Delta \Rightarrow \hat{a} \cap \Delta \neq \emptyset \)
\end{proof}
Returning to the classification \(\mathcal{A}\) in Table 1, \(\bigcap \emptyset = tok(\mathcal{A})\) and \(\bigcup \{\alpha, \beta \} = tok(\mathcal{A}) \) thus, \(\langle \emptyset, \{\alpha, \beta \}\) is a constraint of \(\mathcal{A}\). Similarly, \(\bigcap  \{\alpha, \beta \} = \emptyset = \bigcup \emptyset\); validating \(\langle \{\alpha, \beta \}, \emptyset \rangle\) as a constraint of \(\mathcal{A}\). \par
Moving forward we will consider set theoretical demonstration.
The theory of \(\mathcal{A}\) [Barwise] is \(Th(\mathcal{A})=\{ \langle \alpha,\delta \rangle,\langle \emptyset,\{ \alpha,  \beta \} \rangle,\langle \{ \alpha,  \beta\},\emptyset\rangle\}\). \par
In view of Table 1, the only token of type \(\alpha\) is a, we can observe that a is also of type \(\delta\); here we remove the curly bracket around the singleton to make it readable. This explanation makes  \(\langle \alpha,\delta \rangle\) a sequent constraint supported by \(\mathcal{A}\). 

Given a classification \(\mathcal{A}\) and a set \(\Gamma\)  of types of \(\mathcal{A}\); let \(\bigcap \Gamma=\{x \in tok(\mathcal{A})| \forall \alpha \in \Gamma,x\vDash_\mathcal{A} \alpha\}\) is a subset of \(tok(\mathcal{A})\); with \(\bigcap \emptyset = tok(\mathcal{A})\) and alternatively, \(\bigcup \Gamma =\{x \in tok(\mathcal{A})|\exists \alpha \in \Gamma,x\vDash_\mathcal{A} \alpha\}\). The following proposition gives an algebraic perspective on the validity of a constraint in a classification. 
The computation of a constraint of a classification is not an easy task, specially, it is computational expensive. Fortunately, the logic module of SymPy, a Python library for symbolic mathematics has a function SOPform (Sum of Products form) that will compute and output this theory in disjunction normal form. These two forms; the SOPform and the implicative form have the same semantic / model as we portrait in Table 2. The SymPy computation with input
\begin{equation} \label{eq3}
SOPform([\alpha,\beta,\delta],minterms,dontcares)
\end{equation}
with \(minterms = [[1,0,1],[0,1,1],[0,1,0]]\) and \(dontcares = []\) will return as output
\begin{equation} \label{eq4}
(\beta \& \tilde \alpha) | (\alpha \& \delta \& \tilde \beta) 
\end{equation}
which in mathematical friendly readable form correspond 
\begin{equation} \label{eq5}
(\beta \wedge \neg \alpha) \vee (\alpha \wedge \delta \wedge \neg \beta) 
\end{equation}. \par
For the rest of this document, we will use SymPy for all our computation.For the rest of this document, we will use SymPy for all our computation.
\begin{table}[h!]
    \centering
    \begin{tabular}{|c|c|c|c|c|c|c|c|c|c}
    \hline
     \( \mathcal{A}\)   & \(\alpha\) & \(\beta\)  & \(\delta\) & \(\neg \alpha\) & \(\neg \beta\) & \(\beta \wedge \neg \alpha\) & \(\neg \beta \wedge \alpha \wedge \delta\) & \((\beta \wedge \neg \alpha) \vee (\neg \beta \wedge \alpha \wedge \delta) \)  \\
    \hline
     0 & 1 & 1 & 1 & 0 & 0 & 0 & 0 & 0  \\
    \hline
     0 & 1 & 1 & 0 & 0 & 0 & 0 & 0 & 0 \\
    \hline
     1 & 1 & 0 & 1 & 0 & 1 & 0 & 1 & 1 \\
    \hline
     0 & 1 & 0 & 0 & 0 & 1 & 0 & 0 & 0 \\
    \hline
     1 & 0 & 1 & 1 & 1 & 0 & 1 & 0 & 1 \\
    \hline
     1 & 0 & 1 & 0 & 1 & 0 & 1 & 0 & 1 \\
    \hline
     0 & 0 & 0 & 1 & 1 & 1 & 0 & 0 & 0 \\
    \hline
     0 & 0 & 0 & 0 & 1 & 1 & 0 & 0 & 0 \\
    \hline
    \end{tabular}
    \caption{True table equating theory of a classification table and the output of sum of product from SymPy}
    \label{tab:my_label1}
\end{table}

\subsection{Multi-classification}
Definition: A multiclassification is giving by a tuple \(\mathcal{A}\) = (tok(\(\mathcal{A}\)), typ(\(\mathcal{A}\)), \(\vDash_\mathcal{A}\)) where \(\vDash_\mathcal{A}\) is a function from tok(\(\mathcal{A}\)) \(\times\) typ(\(\mathcal{A}\)) to  \(\mathbb{N}\).\\
Classification is a set of sets: Each token is a set of its types or, equivalently, for each type, the set of tokens that are classified as being of that type. Multi-classification is a set of multisets. We first provide some basic definitions related to multisets.\\
Definition \cite{multiset}:  Let D be a set.A multiset over D is just a pair \(\langle D, f\rangle\), where D is a set and \(f : D \rightarrow \mathbb{N}\) is a function.\\
For \(d \in D\) f(d) is referred to as the multiplicity of d. The set D is the support of the multiset; sometime one ignored the function and referred to a multiset by one capital letter as for example \(D = \{2a, b, 3c\}\), in this case the support is denote by \(\vert D \vert = \{a, b, c\}\).\\
A multi-classification is just a multiset over the cartesian product tok(\(\mathcal{A}\)) \(\times\) typ(\(\mathcal{A}\))
Given a multi-classification \(\mathcal{A} = (tok(\mathcal{A}), typ(\mathcal{A}), \vDash_\mathcal{A})\),  any type \(\alpha\) defined a multiset on the set of tokens as \(\check{\alpha} : typ(\mathcal{A}) \rightarrow \mathbb{N}  with \check{\alpha}(a) = \vDash_\mathcal{A}(a, \alpha)\) The same procedure applied to define multiset on the set of types for any given token. \\
Operations on multisets are the same as on sets:
If  \(\langle D, f\rangle\) and  \(\langle E, g \rangle\) are two multisets, we have the following multiset operation:

Union: \(\langle D, f\rangle \cup \langle E, g \rangle \myeq (h : D \cup E \rightarrow \mathbb{N}\) with \(h(x) = \max(f(x), g(x))\)\\
Intersection: \(\langle D, f\rangle \cap \langle E, g \rangle \myeq (h : D \cap E \rightarrow \mathbb{N} with h(x) = min(f(x), g(x))\)\\
Submultiset: \(\langle D, f\rangle \leq  \langle E, g \rangle \myeq \forall x \in D \cup E \hspace{0.3cm} f(x) \leq\) g(x)\\
Empty multiset: On a domain D the empty multiset is defined as \(\Theta(x) = 0 \hspace{0.3cm} \forall x \in D \)\\
Given a multiclassification \(\mathcal{A} = (tok(\mathcal{A}), typ(\mathcal{A}), \vDash_\mathcal{A}) and a multiset \Gamma of the set typ(\mathcal{A}) of types, this introduce two multisets on the set tok(\mathcal{A})  as follows: \)\\ 
\(\bigvee \Gamma  : tok(\mathcal{A}) \rightarrow \mathbb{N} with \bigvee \Gamma (x) = \max_{\alpha \in \vert \Gamma \vert} \Gamma(\alpha) \times \check{\alpha}(a) for all a in tok(\mathcal{A}).\)\\

\(\bigwedge \Gamma  : tok(\mathcal{A}) \rightarrow \mathbb{N} with \bigwedge \Gamma(x) = \min_{\alpha \in \vert \Gamma \vert} \Gamma(\alpha) \times \check{\alpha}(a) for all a in tok(\mathcal{A}\)). \\

\subsubsection{Example of multi-classification}
With the same set of tokens and types as in Table 1, consider the following table in  Table3;
\begin{table}[ht!]
    \centering
    \begin{tabular}{|c|c|c|c|c|c}
    \hline
      \(\mathcal{A}\)   & \(\alpha\) & \(\beta\) &  \(\delta\) \\
      \hline
     a    &  3 & 2 & 1 \\
     \hline
     b & 0 & 4 & 2 \\
     \hline
     c & 2 & 3 & 0 \\
     \hline
     
    \end{tabular}
    \caption{Example of multi-classification}
    \label{tab:my_label2}
\end{table}
\subsection{Theory of a Multi-classification}
We now carry over the same investigation for the theory of classification to the theory of a multi-classification.\\
A sequent of a multi-classification is a couple \(\langle \Gamma, \Delta \rangle\)  where  \(\Gamma\)  and  \(\Delta\)  are multisets over the set \(typ(\mathcal{A}\) of types. A token \textit{a} satisfies a sequent \(\langle \Gamma, \Delta \rangle\)   if \(\Gamma \leq \hat{a} \Rightarrow \hat{a} \cap \Delta \neq \Theta; \hat{a}\) is a multiset on \(typ(\mathcal{A})\)  defined as  \(\hat{a}(\alpha) = \vDash_\mathcal{A}(a, \alpha)\). 
The same notion of constraint as for classification is apply to multi-classification.\\
A sequent  \(\langle \Gamma, \Delta \rangle \) is a constraint of a multi-classification \(\mathcal{A}\) if it is satisfy by all tokens of \(\mathcal{A}\); i.e., \(\forall a \in tok(\mathcal{A}) \Gamma \leq \hat{a} \Rightarrow \hat{a} \cap \Delta \neq \Theta\)
Developing a theory for this type of classification will be a matter of linear logic. We will not engage in that direction in this paper; instead, we will convert any multi-classification to binary classification and use available tools such as SOPform from sympy\cite{meurer2017sympy} to extract the theory. \\
The conversion from multi classification to binary classification is carry over by duplicating any given type as many times as indicated by its multiplicity.\\
Having this in mind, the multi classification of Table 3 which can equivalently be denoted as a set of multisets by \(\mathcal{A} = \{ a = \{3\alpha, 2\beta, \gamma\}, b = \{ 4\beta, 2\gamma \}, c = \{ 2\alpha, 3\beta \} \} \equiv \{a = \{\alpha, \alpha, \alpha, \beta, \beta, \gamma \}, b = \{ \beta, \beta, \beta, \beta, \gamma, \gamma \}, c = \{ \alpha, \alpha, \beta, \beta, \beta \}\} \). 
With this conversion, the multi-classification table of Table 3 is reduced to the binary classification of Table 4 below.
\begin{table}[ht!]
    \centering
    \begin{tabular}{|c|c|c|c|c|c|c|c|c|c|c}
    \hline
      \(\mathcal{A}\)   & \(\alpha\) & \(\alpha\) & \(\alpha\) &  \(\beta\) & \(\beta\) & \(\beta\) & \(\beta\) &  \(\gamma\) & \(\gamma\) \\
      \hline
     a    &  1 & 1 & 1  & 1 & 1 & 0 & 0 & 1 & 0 \\
     \hline
     b & 0 & 0 & 0 & 1 & 1 & 1 & 1 & 1 & 1 \\
     \hline
     c & 1 & 1 & 0 & 1 & 1 & 1 & 0 & 0 & 0 \\
     \hline
     
    \end{tabular}
    \caption{Binary Classification from multi-classification}
    \label{tab:my_label3}
\end{table}

Using Sympy as we mentioned before, we obtained the theory of the multi-classification with the command 
\begin{center}
SOPform([\(\alpha, \alpha, \alpha, \beta, \beta, \beta, \beta, \gamma, \gamma\)], [[1,1,1,1,1,0,0,1,0],[0,0,0,1,1,1,1,1,1], [1, 1,0,1,1,1,0,0,0]])
\end{center}
which will produce:\\
\begin{center}
\((\tilde \alpha \wedge \beta \wedge \gamma)\)
\end{center}
\section{From Colored Petri Net to Multi-Classification Tables}
A classification as just defined  has no temporal aspect and is essentially a snapshot at a point in time of a system.  To accommodate change in the is-of-type relation,  a classification table is viewing as marking of an elementary colored Petri net.\\
Definition (Petri Net \cite{badouel}). A Petri net is a tuple \((S, T, l, M_0)\)  where:
\begin{itemize}
\item S is a finite set of places;,
\item T is a finite set of transitions, disjoint from S; 
\item \textit{l} is a labelling function such that 
    \begin{itemize}
    \item  for all \(s \in S\), l(s) is the type of s, i.e., a restriction on the tokens it may hold.
    \item for all \(t \in T\), l(t) is the guard of t,
    \item for all \((x, y) \in (S \times T ) \cup (T \times S)\), l(x,y) is the annotation of the arc from x to y and is a multiset of expressions to specify the tokens produced or consumed through the arc;
    \end{itemize}
\item \(M_0\) is the initial marking; a classification table indicating for place which data type is present at that place.
\end{itemize}
In this document we will restrict ourself to a special class of nets where each place can hold only one token of a given type.
\subsection{Example of Colored Petri Net and its Making Graph}
Figure 2 below portrait an example of colored Petri net with three places \(p_1\), \(p_2\) and \(p_3\) and four transitions \textit{a}, \textit{b}, \textit{c} and \textit{d}. 
\begin{figure}[H]
  \includegraphics{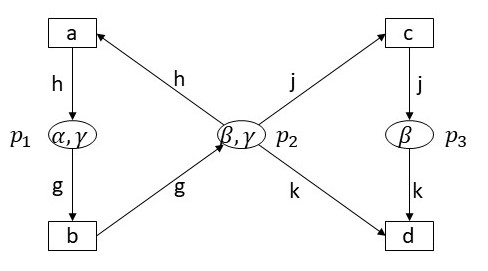}
  \caption{Example of colored Petri net}
\end{figure}
We will represent marking as classification table, and the initial marking of this CPN s it can be viewing in Figure 2 is  
\begin{table}[H]
    \centering
    \begin{tabular}{|c|c|c|c|c|c}
    \hline
     \( M_0 \)   & \(\alpha\) & \(\beta\) &  \(\delta\) \\
      \hline
    \(p_1\)    &  1 & 0 & 1 \\
     \hline
    \( p_2\) & 0 & 1 & 1 \\
     \hline
     \(p_3\) & 0 & 1 & 0 \\
     \hline
     
    \end{tabular}
    \caption{Initial Marking of the Colored Petri Net}
    \label{tab:my_label4}
\end{table}
This colored Petri net is simple enough to make the binding of any transition stand up by itself: At its initial configuration, transition \textit{a} is enabled with binding \(h \equiv \beta \) or \(h \equiv \gamma \); transition \textit{b} is enabled with binding  \(g \equiv \alpha \) or \(h \equiv \gamma \); transition c is enabled with binding  \(j \equiv \beta \) or \(j \equiv \gamma \); and, transition d is enabled with binding  \(k \equiv \beta \).\\
For some particular application, there is sometime a constraint not having more than one token of a given type in a place. If one is reinforcing this type of constraint, not all binding mentioned above are enabled. In this context transition \textit{a} is enabled with only binding \(h \equiv \beta \); the binding \(h \equiv \gamma \) cannot longer be fufill because place \(p_1\) has already a token of type \(\gamma\). Similarly, transition \textit{b} is enabled with only the binding \(g \equiv \alpha \), transition \textit{c} is enabled with the binding \(j \equiv \gamma \) and transition is not enabled at all. In general, at the initial marking \(M_0\), all the transition are enabled, if by a mechanism that we will not go through in this document the colored Petri net moves to fire transition \textit{a} with binding \(h \equiv \beta \), the next configuration will be the classification \(M_1\) with
\begin{table}[H]
    \centering
    \begin{tabular}{|c|c|c|c|c|c}
    \hline
     \( M_1 \)   & \(\alpha\) & \(\beta\) &  \(\delta\) \\
      \hline
    \(p_1\)    &  1 & 1 & 1 \\
     \hline
    \( p_2\) & 0 & 0 & 1 \\
     \hline
     \(p_3\) & 0 & 1 & 0 \\
     \hline
     
    \end{tabular}
    \caption{Next Marking after the initial marking}
    \label{tab:my_label5}
\end{table}
In the context of at most one token of any given type in a place, a quick combinatorics computation shows that there will be 27 makings (classifications) reachable. Although, in general, the behavior of a CPN is non-deterministic \cite{Obaidat} (from one marking, there are multiple reachable markings), we will use SNAKES \cite{Pommereau} a general-purpose Petri nets library to generate all marking of our colored Petri net. 
\begin{figure}[H]
  \includegraphics{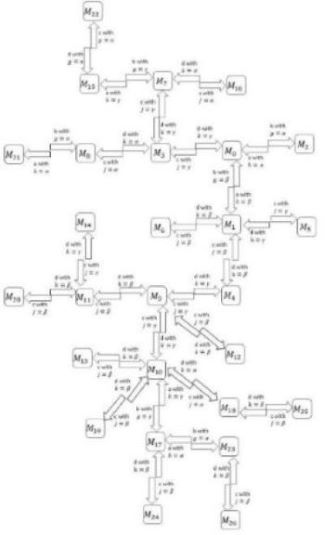}
  \caption{A state graph of a colored Petri net}
\end{figure} 
A state graph of a Petri net is the graph of markings where nodes or vertices are vectors, states, or configurations of the Petri net as enabled transitions are fired (in our case of CPN, configurations are multi-classifications). Labeled edges are transition that is being fired. For our CPN of Figure 2, we have a total of 108 markings. 
The reduced state graph  where places can contains at most one token of a given type, is represent in Figure 3 above.\\
\subsection{Theory of a Colored Petri Net}
The state graph of a CPN is given here with nodes as multi-classification, and the edge between two nodes is transition fired with appropriate binding. We combine the theories extracted from each multi-classification to form a knowledge base of the CPN.
For example amalgamating the theories of multi-classification obtained from CPN of Figure 2
gives us the following knowledge base:\\

\(CPN-KB = \{(\neg a \wedge \neg b \wedge c) \vee (\neg a \wedge \neg b \wedge \neg c),\\
(a \wedge \neg b \wedge c) \vee (\neg a \wedge b \wedge \neg c),\\
( \neg b \wedge c) \vee (\neg a \wedge b \wedge \neg c),\\
(\neg a \wedge c) \vee (a \wedge \neg b  \neg c),\\
(a \wedge \neg b \wedge \neg c) \vee (\neg a \wedge b \wedge \neg c),\\
(\neg a \wedge b ) \vee ( b \wedge \neg c),\\
(\neg a \wedge  b \wedge c) \vee ( a \wedge \neg b \wedge \neg c),\\
(\neg a \wedge  b \wedge c) \vee (\neg a \wedge \neg b \wedge \neg c),\\
(b \wedge \neg c) \vee (\neg a \wedge \neg b \wedge c),\\
(\neg a \wedge  b \wedge c) \vee (a \wedge \neg b \wedge \neg c)\vee (\neg a \wedge \neg b \wedge  c),\\
(a \wedge  b \wedge c) \vee (\neg a \wedge \neg b \wedge \neg c),\\
(a \wedge  b \wedge c) \vee (\neg a \wedge  \neg c),\\
(a \wedge  b \wedge \neg c) \vee ( \neg a \wedge b \wedge  c)\vee (\neg a \wedge \neg b \wedge \neg c),\\
(\neg a \wedge  c) \vee (\neg b \wedge c),\\
(a \wedge  b \wedge \neg c) \vee ( \neg a \wedge  \neg b \wedge  c)\vee (\neg a \wedge \neg b \wedge \neg c),\\
(a \wedge \neg b \wedge \neg c) \vee ( \neg a \wedge b \wedge \neg c)\vee (\neg a \wedge \neg b \wedge c),\\
(\neg b \wedge \neg c) \vee (\neg a \wedge b \wedge c),\\
(a \wedge \neg b \wedge c) \vee ( \neg a \wedge b \wedge  c)\vee (\neg a \wedge \neg b \wedge \neg c),\\
(\neg a \wedge b) \vee (a \wedge \neg  b \wedge c),\\
(a \wedge  b \wedge c) \vee ( \neg a \wedge b \wedge  \neg c) \vee (\neg a \wedge \neg b \wedge \neg c),\\
( b \wedge c) \vee (\neg a \wedge \neg  b \wedge \neg c),\\
(a \wedge b \wedge \neg c) \vee (\neg a \wedge \neg  b \wedge  c),\\
(a \wedge  b \wedge  \neg c) \vee ( \neg a \wedge b \wedge  c) \vee (\neg a \wedge \neg b \wedge \neg c),\\
( \neg a \wedge b) \vee (\neg a \wedge  c),\\
(\neg a \wedge  b \wedge c) \vee ( a \wedge \neg b \wedge \neg c) \vee (\neg a \wedge b \wedge \neg c),\\
( \neg a \wedge b \wedge c) \vee (\neg a \wedge \neg b)\}\)

\section{Conclusion and Futures Research}
In this research, we expand the notion of classification table as defined by Barwise and Seligman by defining a type as a multiset of tokens in the general sense. We then show how to extract its theory with the same tool for classification. The main result of this paper is the equivalence we established between a multi-classification and the marking of a CPN. With this observation, we defined as regularities of a CPN the amalgamation of the theory of its markings.\\
We extracted the theory of a multi-classification by reducing it to a classification table and using Sympy. As future work, we intend to study this theory on its own in the light of linear logic -- linear sequent calculus.
\printbibliography

\end{document}